\documentclass[12pt]{article}
\usepackage[authoryear]{natbib}
\usepackage{amssymb}
\usepackage{amsfonts}
\usepackage{amsmath}
\usepackage{dsfont}
\usepackage{soul}
\usepackage[nohead]{geometry}
\usepackage{graphicx}
\usepackage{amsthm}
\usepackage{color}
\usepackage{comment}
\usepackage{setspace}
\usepackage{framed}
\usepackage{enumitem}
\usepackage{todonotes}
\usepackage{tikz}
\usepackage{rotating}
\usepackage{subfigure}
\usepackage[flushleft]{threeparttable}
\usepackage{multirow}
\usepackage[colorlinks=true,citecolor=blue,linkcolor=blue]{hyperref}
\usepackage{xr}
\usepackage{bm}
\usepackage{xcolor}
\usepackage{booktabs}
\usepackage{lscape}
\usepackage{algorithm}
\usepackage{algorithmic}
\usepackage[toc,title,page]{appendix}
\usepackage{lscape}
\usepackage{pdflscape}

\usepackage{natbib}

\renewcommand{\today}{\ifcase \month \or January\or February\or March\or %
	April\or May\or June\or July\or August\or September\or October\or November\or %
	December\fi, \number \year} 

\usepackage{caption} 
\allowdisplaybreaks

7

\newcommand{\sgn}{\mathrm{sgn}}

\newcommand{\bA}{\mbox{\bf A}}

\newcommand{\bB}{\mbox{\bf B}}

\newcommand{\bD}{\mbox{\bf D}}

\newcommand{\bH}{\mbox{\bf H}}

\newcommand{\bW}{\mbox{\bf W}}
\newcommand{\bI}{\mbox{\bf I}}
\newcommand{\bJ}{\mbox{\bf J}}
\newcommand{\bV}{\mbox{\bf V}}

\newcommand{\bR}{\mbox{\bf R}}

\newcommand{\bU}{\mbox{\bf U}}

\newcommand{\bGamma}{\mbox{\boldmath $\Gamma$}}
\newcommand{\bUpsilon}{\mbox{\boldmath $\Upsilon$}}
\newcommand{\bXi}{\mbox{\boldmath $\Xi$}}
\newcommand{\bLambda}{\mbox{\boldmath $\Lambda$}}
\newcommand{\bTheta}{\mbox{\boldmath $\Theta$}}

\newcommand{\bSigma}{\mbox{\boldmath $\Sigma$}}
\newcommand{\bOmega}{\mbox{\boldmath $\Omega$}}

\newcommand{\bPhi}{\mbox{\boldmath $\Phi$}}
\newcommand{\bPsi}{\mbox{\boldmath $\Psi$}}

\newcommand{\cov}{\mathrm{cov}}

\newcommand{\tr}{\mathrm{tr}}
\newcommand{\diag}{\mathrm{diag}}

\newcommand{\beq}{\begin{eqnarray*}}
\newcommand{\eeq}{\end{eqnarray*}}

\usepackage{graphicx,rotating,booktabs,threeparttable}

\usepackage{adjustbox}
\usepackage{array}

\newcolumntype{R}[2]{%
	>{\adjustbox{angle=#1,lap=\width-(#2)}\bgroup}%
	l%
	<{\egroup}%
}

\newtheorem{thm}{Theorem}[section]

\newtheorem{lem}{Lemma}[section]

\newtheorem{assum}{Assumption}[section]

\numberwithin{equation}{section}
\theoremstyle{definition}

\newtheorem{remark}{Remark}[section]
\makeatletter
\def\@biblabel#1{\hspace*{-\labelsep}}
\makeatother
\begin{document}
	\bibliographystyle{ecta}
	
	\title{Large Global Volatility Matrix Analysis Based on Observation Structural Information}

	\date{\today}

	\author{
		Sung Hoon Choi\thanks{%
			Department of Economics, University of Connecticut, Storrs, CT 06269, USA. 
			E-mail: \texttt{sung\_hoon.choi@uconn.edu}.} \\ 
		\and Donggyu Kim\thanks{College of Business, Korea Advanced Institute of Science and Technology (KAIST), Seoul, Republic of Korea.
			Email: \texttt{donggyukim@kaist.ac.kr}.}\\ 
		}
	\maketitle
	\pagenumbering{arabic}	
\begin{abstract}
		\onehalfspacing
		In this paper, we develop a novel large volatility matrix estimation procedure for analyzing global financial markets. 
		Practitioners often use lower-frequency data, such as weekly or monthly returns, to address the issue of different trading hours in the international financial market.
		However, this approach can lead to inefficiency due to information loss. 
		To mitigate this problem, our proposed method, called Structured Principal Orthogonal complEment Thresholding (Structured-POET), incorporates observation structural information for both global and national factor models. 
		We establish the asymptotic properties of the Structured-POET estimator, and also demonstrate the drawbacks of conventional covariance matrix estimation procedures when using lower-frequency data. 
		Finally, we apply the Structured-POET estimator to an out-of-sample portfolio allocation study using international stock market data.\\
		
\noindent \textbf{Key words:} High-dimensionality, international financial market, low-rank matrix, multi-level factor model, POET, sparsity.
					
\end{abstract}

\newpage
	
\doublespacing

\section{Introduction}

Factor analysis and principal component analysis (PCA) are commonly used for various applications, including macroeconomic variable forecasting and portfolio allocation optimization \citep{bai2003inferential, bernanke2005measuring, fan2016incorporating, stock2002forecasting}.
Recent research has highlighted the importance of considering local factors in addition to global factors.
These local factors have an impact on individuals in each local group and can be defined by regional, country, or industry level \citep{bekaert2009international, chaieb2021factors, fama2012size, kose2003international, moench2013dynamic}.
 To account for different level factors, a multi-level factor structure has been developed.
See \cite{ando2015asset, ando2016panel, bai2015identification,  choi2018multilevel, han2021shrinkage} for more information.

Several large volatility matrix estimation procedures have been developed based on latent factor models to account for the strong cross-sectional correlation in the stock market \citep{ait2017using, fan2016incorporating, fan2018large, fan2019structured, jung2022next, wang2017asymptotics}. 
For instance, under the single-level factor model, \cite{fan2013large} proposed a principal orthogonal component thresholding (POET) covariance matrix estimation procedure when the factors are unobservable. 
This method can consistently estimate unobservable factors using PCA and cross-sectionally correlated errors via thresholding with a large number of assets.
Recently, to account for the latent local factor structure, \cite{choi2023large} developed a Double-POET covariance matrix estimation procedure based on the multi-level factor structure. 
Specifically, Double-POET is a two-step estimation procedure that estimates global and national factors separately by applying PCA at each factor level based on the block local factor structure.

The analysis of international financial markets is crucial for constructing global benchmark indices, such as the MSCI World.
When analyzing global financial market data, it is a common practice to use lower frequencies, such as two-day average, weekly, monthly, or quarterly data, instead of daily returns \citep{ando2017clustering, bekaert2009international, chib2006analysis, choi2023large, fama2012size, hou2011factors}. 
This is because international stock markets operate at different times, which results in returns being based on different information sets when measured on any given date over a short period, such as daily.
Hence, practitioners often opt for using weekly or monthly returns to mitigate the impact of non-synchronized trading hours in international markets. 
\cite{choi2023large} also used weekly returns to analyze the large global volatility matrix. 
However, using lower-frequency data can lead to a loss of information and less efficient estimators.
On the other hand, in high-frequency financial econometrics literature, several papers considered asynchronous financial data for covariance estimations by sample splitting method \citep{ait2010high, dai2019knowing, hautsch2012blocking, lunde2016econometric}.
See also \cite{burns1998correlations,  fan2019structured,  sun2022factor} for related articles.
However, unlike a high-frequency model setup, non-synchronized low-frequency observations can cause not only inefficiency but also inconsistency because a observation time gap is fixed.
Thus, it is important to study the non-synchronized trading hours in international markets based on low-frequency observations and develop an efficient and effective large global volatility matrix estimation procedure.

This paper proposes a novel large volatility matrix estimation procedure that incorporates observation structural information with an entire set of observations in the global and national factor model. 
Specifically, we consider the international financial market and impose a latent multi-level factor structure to account for both global and national risk factors. 
To handle the non-synchronized trading hours in the international market, we assume that the correlation is stationary with respect to the relative trading hour difference.
For example, if the proportion of the overlapped time goes to one, its corresponding correlation converges to the correlation of the synchronized time, which is the parameter of interest. 
This structure helps theoretically understand the non-synchronized trading hour problem in the international market, and we study the large global volatility matrix estimation problem under the proposed stationary  global and national factor model. 
For example, national factor membership is assumed to be naturally known, and we further assume that countries within the same continent have the same information set.
To estimate the global volatility matrix, we first apply the Double-POET procedure in each continental group using daily returns, which have the finest information in our model setup.
 To capture the spillover effect between continents, we conduct a low-rank approximation to each continental pair using a lower-frequency, which helps mitigate the effect of the non-synchronized trading hours. 
Finally, to accommodate the latent global factors, we employ PCA on the structurally fitted global factor components from previous procedures, which we call Structured Principal Orthogonal complEment Thresholding (Structured-POET). 
We derive the rates of convergence for the Structured-POET estimator under different matrix norms and discuss its benefits compared to the Double-POET procedure. 
For example, the Double-POET estimator may not be consistent using daily returns, while the Double-POET estimator with lower-frequency returns has a slower convergence rate than the Structured-POET estimator.
An empirical study on portfolio allocation also supports the theoretical findings and shows the benefit of the proposed Structured-POET estimator.

The remainder of the paper is organized as follows. 
In Section \ref{section2}, we introduce the model and propose the Structured-POET estimation procedure. 
Section \ref{asymp} provides an asymptotic analysis of the Structured-POET estimator. 
In Section \ref{simulation}, we conduct a simulation study to evaluate the finite sample performance of the proposed method. 
Section \ref{empiric} applies the proposed method to a real data problem of portfolio allocation using global stock market data. 
Finally, the conclusion is provided in Section \ref{conclusion}. 
All proofs and miscellaneous materials are presented in the online supplementary file.

\section{Model Setup and Estimation Procedure} \label{section2}

		Let $\lambda_{\min}(\bA)$ and $\lambda_{\max}(\bA)$ denote the minimum and maximum eigenvalues of matrix $\bA$, respectively. 
		 In addition, we denote by $\|\bA\|_{F}$, $\|\bA\|_{2}$ (or $\|\bA\|$ for short), $\|\bA\|_{1}$, $\|\bA\|_{\infty}$, and $\|\bA\|_{\max}$ the Frobenius norm,  operator norm, $l_{1}$-norm, $l_{\infty}$-norm and elementwise norm, which are defined, respectively, as $\|\bA\|_{F} = \tr^{1/2}(\bA'\bA)$, $\|\bA\|_{2} = \lambda_{\max}^{1/2}(\bA'\bA)$, $\|\bA\|_{1} = \max_{j}\sum_{i}|a_{ij}|$, $\|\bA\|_{\infty} = \max_{i}\sum_{j}|a_{ij}|$, and $\|\bA\|_{\max} = \max_{i,j}|a_{ij}|$. 
		 When $\bA$ is a vector, the maximum norm is denoted as $\|\bA\|_{\infty}=\max_{i}|a_{i}|$, and both $\|\bA\|$ and $\|\bA\|_{F}$ are equal to the Euclidean norm. We denote $\diag(\bA_{1},\ldots, \bA_{n})$ with the diagonal block entries as $\bA_{1},\ldots, \bA_{n}$.

\subsection{Multi-Level Factor Model} \label{setup}
We consider a global and national factor model \citep{choi2023large}:
			\begin{equation}\label{model}
				y_{it} = b_{i}'G_{t} + {\lambda_{i}^{l}}'f_{t}^{l}+ u_{it},  \text{ for } i = 1,\dots, p,  t = 1,\dots,T, \text{ and } l = 1, \dots, L,
			\end{equation}
where $y_{it}$ is the observed data for the $i$th individual belonging to country $l$ at time $t$;
$G_{t}$ is a $k \times 1$ vector of unobserved latent global factors, and $b_{i}$ is the global factor loadings; $f_{t}^{l}$ is an $r_{l} \times 1$ vector of unobserved latent national factors that affect individuals belonging to  country $l$, and $\lambda_{i}^{l}$ is the corresponding national factor loadings; and  $u_{it}$ is an idiosyncratic error term, which is uncorrelated with $G_{t}$ and $f_{t}^{l}$.  
Throughout the paper, global and national factors are uncorrelated, while their factor loadings may not be orthogonal to each other. 
In addition, we assume that the group membership of the national factors is known and the numbers of factors, $k$ and $r_l$, are fixed.

Due to the known national factor membership, we can stack the observations and write the model (\ref{model}) in a vector form as follows:
			\begin{equation}\label{model_vector}
			y_{t} = \bB G_{t} + \bLambda F_{t} +u_{t},
			\end{equation}
where  $y_{t} = ({y_{t}^{1}}',\dots, {y_{t}^{L}}')'$, where $y_{t}^l = (y_{(\sum_{j=0}^{l-1}p_{j}+1)t}, \dots, y_{(\sum_{j=0}^{l}p_{j})t})'$, $p_l$ is the number of assets within country $l$, and $p_0 =0$; the $p\times k$ matrix $\bB = (b_1, \dots, b_p)'$; the $p \times r$ block diagonal matrix $\bLambda = \diag(\bLambda^1,  \dots, \bLambda^L)$, where $\bLambda^{l} = (\lambda_{1}^{l}, \dots, \lambda_{p_{l}}^{l})'$ is a $p_l \times r_l$ matrix of local factor loadings for each $l$ such that $r = \sum_{l=1}^{L} r_l$; the $r \times 1$ vector $F_{t} = {(f_{t}^{1}}', \dots, {f_{t}^{l}}')'$, where $f_{t}^{l}$ is an $r_l \times 1$ vector of local factors; and $u_{t} = (u_{1t}, \dots, u_{pt})'$.
	
In this paper, we are interested in the $p \times p$ covariance matrix of $y_{t}$ and its inverse matrix:
			\begin{equation}\label{model2}
			\bSigma = \bB \cov(G_{t})\bB' + \bLambda \cov(F_{t})\bLambda' + \bSigma_{u} := \bSigma_{g} + \bSigma_{l} + \bSigma_{u},
			\end{equation}
where $\bSigma_{u}$ is a sparse idiosyncratic covariance matrix of $u_{t}$.
In particular, we measure the sparsity level of $\bSigma_{u} = (\sigma_{u,ij})_{p\times p}$ as follows \citep{bickel2008covariance, cai2011adaptive, gagliardini2016time, gagliardini2020estimation, rothman2009generalized}: for some $q \in [0,1)$,
			\begin{align} \label{sparsity}
				m_{p} = \max_{i\leq p}\sum_{j\leq p} |\sigma_{u,ij}|^{q}, 
			\end{align}
which diverges slowly, such as $\log p$. 
This implies that most pairs are weakly cross-sectionally correlated in the idiosyncratic error component. 
Decomposition \eqref{model2} is a multi-level factor-based covariance matrix.
Under the pervasive assumption, there are distinguished eigenvalues among the global factor components, the local factor components, and the idiosyncratic error components.
Hence, we can analyze the model by the presence of distinguished eigenvalues at different levels (see Section \ref{estimation procedure}).
We note that the correlation matrix of $y_{t}$ can be obtained by
			\begin{equation}\label{def-cor}
				\bR_{0}  = (\rho_{0,ij})_{p\times p} = \bD_{0}^{-\frac{1}{2}}\bSigma \bD_{0}^{-\frac{1}{2}},
			\end{equation}		
where $\bD_{0}$ is the diagonal matrix consisting of the diagonal elements of $\bSigma$.
Importantly, the correlation between stocks $i$ and $j$, denoted by $\rho_{0,ij}$ in \eqref{def-cor}, can be realized through synchronized observations.
However, in the context of the international stock market, each stock exchange operates its own trading hours. 
Hence, stocks traded on different exchanges have distinct observation time points, making it challenging to estimate $\rho_{0,ij}$ if stocks $i$ and $j$ are not in the same region.
To address this issue, we introduce the following model structure.
The $i$th observations in region $s$ is $\{y_{i,t+\delta_{s}}\}_{t=1}^{T}$, and $\delta_{s}\in [0,1)$ is the market close time for $s \in \{1,\dots,S\}$.
We assume that $\{(y_{1,t+\delta_{1}},\dots, y_{p,t+\delta_{S}})'\}_{t\geq 1}$ is stationary.
Denote the estimable correlation by $\rho_{h,ij}$ for assets $i$ and $j$ that are located in regions $s, q \in \{1,\dots, S\}$, respectively, where  the relative time difference $h= \frac{|\delta_{s} - \delta_{q}|}{d}$ and the window size of frequency $d = T^{1-\alpha}$ for $\alpha \in (0,1]$.
 Finally, we impose the following Lipschitz condition for $\rho_{h,ij}$: 
\begin{equation}
	|\rho_{h,ij} - \rho_{0,ij}| \leq Ch^{\beta}, \label{lipschitz}
\end{equation}
for some $\beta > 0$ and a positive constant $C$.  
This model setup provides a mathematical framework to understand a fraction of the global market.
For example, from the proposed model setup perspective, when using daily return data, $\rho_{h,ij}$ does not converge to the synchronized correlation $\rho_{0,ij}$.
In contrast, when using lower-frequency data, $\rho_{h,ij}$ converges to  $\rho_{0,ij}$.
Thus, in practice, researchers often use weekly or monthly returns instead of daily returns to mitigate the effect of different trading hours based on daily transaction prices \citep{ando2017clustering, bekaert2009international, chib2006analysis, fama2012size, hou2011factors}. 
However, this causes inefficiency.
We discuss this inefficiency theoretically in Section \ref{asymp}.

In this paper, for simplicity, we assume that stocks in the same continent have the same observation time points; hence, regional membership is the continent. 
Naturally, regional membership is known.
In addition, we assume that the number of regions, $S$, is fixed.
Given the regional membership, we can stack the observations by country within each continent.  
Then, we define the ``estimable" correlation matrix as follows:
			\begin{align}\label{R_h}
				\bR_{h} = \begin{bmatrix}
					\bR_{0,11} & \bR_{h,12} & \cdots & \bR_{h,1S} \\
					\bR_{h,21} & \bR_{0,22} & \cdots & \bR_{h,2S} \\
					\vdots          & \vdots          & \ddots & \vdots          \\
					\bR_{h,S1} & \bR_{h,S2} & \cdots & \bR_{0,SS}
				\end{bmatrix},
			\end{align}			
where $\bR_{0,ss} =(\rho_{0,ij})_{p_{s} \times p_{s}}$ and $\bR_{h,sq} =(\rho_{h,ij})_{p_{s} \times p_{q}}$ for  $s, q \in \{1,\dots, S\}$. 
For simplicity, we use the subscript notation $h$, which is a function of $i$, $j$, and $d$.
We note that, for $s\neq q$, $\bR_{h,sq}$ represents the spillover effect between continents $s$ and $q$, and we assume that its rank is  $k^{*}_{sq}$.
Moreover, without loss of generality, each rank is at most equal to the number of global factors (i.e., $k^{*}_{sq} \leq k$).
We denote the corresponding covariance matrix by $\bSigma_{h} = \bD_{0}^{\frac{1}{2}} \bR_{h} \bD_{0}^{\frac{1}{2}}$.

Let $\bSigma^{s}$ be the covariance matrix for continent $s$, which is a $p_{s} \times p_{s}$ diagonal block of $\bSigma$.
We then decompose $\bSigma^{s}$ as follows:
			\begin{equation}\label{continent cov}
				\bSigma^{s} = \bSigma_{g}^{s} + \bSigma_{l}^{s} + \bSigma_{u}^{s}, \;\;\;\; \text{ for } s = 1, \dots, S, 
			\end{equation}
where $\bSigma_{g}^{s}$ is the global factor component, $\bSigma_{l}^{s}$ is the national factor component, and $\bSigma_{u}^{s}$ is the idiosyncratic component.
The equation \eqref{continent cov} represents the multi-level factor-based covariance matrix.
Thus, when we consider markets that have the same observation time, we can directly apply the Double-POET procedure proposed by \cite{choi2023large} with all possible observations (i.e., $d=1$) to estimate the covariance matrix $\bSigma^{s}$ (see Section A.1 in the online supplement for the specific estimation procedure).
However, when analyzing the global stock market,  lower-frequency data is often used to avoid the issue of varying trading hours,  which causes inefficiency.
To handle this issue, we propose a novel procedure for estimating large global covariance matrix $\bSigma$ in the following subsection, which incorporates the structural information with an entire set of observations.

\subsection{Structured-POET Procedure}  \label{estimation procedure}

Following  \cite{choi2023large}, we assume the canonical conditions that $\cov(G_{t}) = \bI_{k}$, $\cov(f_{t}^{l}) = \bI_{r_l}$, and  $\bB'\bB$ and ${\bLambda^{l}}'\bLambda^{l}$  are diagonal matrices for $l \in \{1, \dots, L\}$. 
We note that $G_{t}$ and $f_{t}^{l}$ are uncorrelated with each other.
Let $a_{1}, a_{2} \in (0,1]$ be the strengths of global and local factors, respectively.
We then impose the following pervasiveness conditions: for each $l$, the eigenvalues of $p^{-a_{1}}\bB'\bB$ and $p_{l}^{-a_{2}}{\bLambda^{l}}'\bLambda^{l}$ are distinct and bounded away from zero.
This condition implies that the first $k$ eigenvalues of $\bB \cov(G_{t})\bB'$ diverge at rate $O(p^{a_{1}})$, while the first $r$ eigenvalues of $ \bLambda \cov(F_{t})\bLambda'$ diverge at rate $O(p^{ca_{2}})$, where $a_{1} > ca_{2}$.
For $c \in (0,1]$, we note that $p^c \asymp p_{l}$ for each country $l$.
Also, all eigenvalues of $\bSigma_u$ are bounded.

To incorporate the structure of the global financial market discussed in Section \ref{setup}, we propose a Structured-POET procedure to estimate $\bSigma$ as follows:
\begin{enumerate}
\item For each continent $s$, we compute the Double-POET estimator \citep{choi2023large} using $T$ observations and denote it as $\widehat{\bSigma}^{s,\mathcal{D}} \equiv \widehat{\bSigma}_{g}^{s,\mathcal{D}} + \widehat{\bSigma}_{l}^{s,\mathcal{D}} + \widehat{\bSigma}_{u}^{s,\mathcal{D}}$. 
The specific procedure is described in Appendix A.1. 
Let $\widetilde{\bSigma}_{g}^\mathcal{D} = \diag(\widehat{\bSigma}_{g}^{1,\mathcal{D}},\dots,\widehat{\bSigma}_{g}^{S,\mathcal{D}})$, $\widetilde{\bSigma}_{l}^\mathcal{D} = \diag(\widehat{\bSigma}_{l}^{1,\mathcal{D}},\dots,\widehat{\bSigma}_{l}^{S,\mathcal{D}})$, and $\widetilde{\bSigma}_{u}^\mathcal{D} = \diag(\widehat{\bSigma}_{u}^{1,\mathcal{D}},\dots,\widehat{\bSigma}_{u}^{S,\mathcal{D}})$.
Then, we construct a block diagonal matrix
\begin{equation}\label{diag_cov}
		\widetilde{\bSigma}^\mathcal{D} = \diag(\widehat{\bSigma}^{1,\mathcal{D}},\dots,\widehat{\bSigma}^{S,\mathcal{D}}) \equiv \widetilde{\bSigma}_{g}^\mathcal{D} + \widetilde{\bSigma}_{l}^\mathcal{D} +\widetilde{\bSigma}_{u}^\mathcal{D}.
\end{equation}
\item Given a sample covariance matrix using $d$-day return data, $\widehat{\bSigma}_{h} = T^{-\alpha}\sum_{t=1}^{T^{\alpha}}(y_{t} - \bar{y})(y_{t} - \bar{y})'$, we compute the sample correlation matrix $\widehat{\bR}_{h} = \widehat{\bD}_{h}^{-\frac{1}{2}}\widehat{\bSigma}_{h} \widehat{\bD}_{h}^{-\frac{1}{2}}$, where $\widehat{\bD}_{h}$ is the diagonal matrix consisting of the diagonal elements of $\widehat{\bSigma}_{h}$. 
We denote the sample correlation matrix $\widehat{\bR}_{h}$ as the following block matrix form:
		$$ 
		\widehat{\bR}_{h} = (\widehat{\rho}_{h,ij})_{p\times p} = \begin{bmatrix}
			\widehat{\bR}_{h,11} & \widehat{\bR}_{h,12} & \cdots & \widehat{\bR}_{h,1S} \\
			\widehat{\bR}_{h,21} & \widehat{\bR}_{h,22} & \cdots & \widehat{\bR}_{h,2S} \\
			\vdots          & \vdots          & \ddots & \vdots          \\
			\widehat{\bR}_{h,S1} & \widehat{\bR}_{h,S2} & \cdots & \widehat{\bR}_{h,SS}
		\end{bmatrix}.
		$$
For each $(s,q)$th off-diagonal partitioned block, we conduct the best rank-$k^{\ast}_{sq}$ matrix approximation  to $\widehat{\bR}_{h,sq}$ such that $\widehat{\bTheta}_{sq} = \sum_{i=1}^{k^{*}_{sq}}\widehat{\xi}_{i}\widehat{u}_{i}\widehat{w}_{i}^{\prime}$, where $\{\widehat{\xi}_{i}, \widehat{u}_{i}, \widehat{w}_{i}\}_{i=1}^{p_s \wedge p_q}$ are the ordered singular values, left-singular and right-singular vectors of $\widehat{\bR}_{h,sq}$ in decreasing order.
Then, we define		
		$$
		\widehat{\bTheta} = \begin{bmatrix}
			\mathbf{0} & \widehat{\bTheta}_{12} & \cdots & \widehat{\bTheta}_{1S} \\
			\widehat{\bTheta}_{21} & \mathbf{0}  & \cdots & \widehat{\bTheta}_{2S} \\
			\vdots          & \vdots          & \ddots & \vdots          \\
			\widehat{\bTheta}_{S1} & \widehat{\bTheta}_{S2} & \cdots & \mathbf{0} 
		\end{bmatrix}.
		$$
		
\item  Let $\widetilde{\delta}_{1} \geq \widetilde{\delta}_{2} \geq \dots \geq \widetilde{\delta}_{k}$ be the $k$ largest eigenvalues  of $\widetilde{\bSigma}_{g} = (\widetilde{\bSigma}_{g}^\mathcal{D} +  \widehat{\bD}^{\frac{1}{2}}\widehat{\bTheta} \widehat{\bD}^{\frac{1}{2}})$ and $\{\widetilde{v}_{i}\}_{i=1}^{k}$ be their corresponding eigenvectors, where $\widehat{\bD}$ is the diagonal matrix consisting of the diagonal elements of \eqref{diag_cov}.
The final estimator of $\bSigma$ is then defined as
		\begin{equation}\label{estimator}
			\widehat{\bSigma}^{\mathcal{S}} =  \widetilde{\bV}_{g}\widetilde{\bGamma}_{g}\widetilde{\bV}_{g}' +  \widehat{\bSigma}_{E}^\mathcal{S},
		\end{equation}
where $\widetilde{\bGamma}_{g} = \diag(\widetilde{\delta}_{1}, \dots, \widetilde{\delta}_{k})$, $\widetilde{\bV}_{g}=(\widetilde{v}_{1},\dots, \widetilde{v}_{k})$, $\widehat{\bSigma}_{E}^{\mathcal{S}} = \widetilde{\bSigma}_{l}^\mathcal{D} +\widetilde{\bSigma}_{u}^\mathcal{D}$, and $\widetilde{\bSigma}_{l}^\mathcal{D}$ and $\widetilde{\bSigma}_{u}^\mathcal{D}$ are defined in \eqref{diag_cov}. 		
\end{enumerate}

\begin{remark}\label{rankchoice}
	To implement Structured-POET, we need to determine the rank $k^{*}_{sq}$ and the number of factors, which are unknown in practice.
	We note that each $(s,q)$th off-diagonal partitioned block $\bR_{h,sq}$ in \eqref{R_h} is a low-rank matrix, and each rank is less than or equal to the number of global factors (i.e., $k^{*}_{sq} \leq k$).
	Thus, to determine the rank and number of global factors, we can use the data-driven methods proposed by \citet{ahn2013eigenvalue, bai2002determining, fortin2023, gagliardini2019diagnostic, onatski2010determining}.
	For example, the rank $k^{*}_{sq}$ can be determined by finding the largest singular value gap such that $\max_{i \leq \bar{k}_{sq}} (\widehat{\xi}_{i} - \widehat{\xi}_{i+1})$, where $\bar{k}_{sq} = \min\{p_{s}, p_{q}\}$.
	On the other hand, to consistently estimate $k$, we employ the modified version of the eigenvalue ratio method, introduced by \cite{choi2023large}, based on $\widehat{\bSigma}_{h}$.
\end{remark}

This procedure is called Structured-POET.
Structured-POET efficiently estimates global and local factor components by utilizing all observations and considering the block structure of the local factors.
In contrast, the Double-POET or POET estimator may not be efficient due to the loss of the observation structural information.
In Section \ref{asymp}, we discuss the theoretical inefficiency of Double-POET.
Also, the numerical studies in Sections \ref{simulation} and \ref{empiric} show that Structured-POET outperforms Double-POET and POET.

\section{Asymptotic Properties} \label{asymp}

This section establishes the asymptotic properties of the proposed Structured-POET estimator. 
To investigate asymptotic behaviors, we require the following technical assumption.

\begin{assum} \label{assum1} ~
\begin{itemize}
\item[(i)]  $\cov(G_{t}) = \bI_{k}$, $\cov(f_{t}^{l}) = \bI_{r_l}$, and  $\bB'\bB$ and ${\bLambda^{l}}'\bLambda^{l}$  are diagonal matrices for $l \in \{1, \dots, L\}$. In addition, $G_{t}$ and $f_{t}^{l}$ are uncorrelated with each other.
\item[(ii)] 
For some constants $c \in (0,1]$, $a_{1} \in (\frac{3+2c}{5},1]$, and $a_{2} \in (\frac{3}{5},1]$, all  eigenvalues of $\bB'\bB/p^{a_{1}}$ and $\bLambda^{l\prime}\bLambda^{l}/p_{l}^{a_{2}}$ are strictly bigger than zero as $p, p_{l} \rightarrow \infty$, for $l \in \{1,\dots,L\}$.
In addition, $p_{l} \asymp p^{c}$, for each country $l$, and $a_{1}\geq ca_{2}$.
There is a constant $C>0$ such that $\|\bB\|_{\max}\leq C$ and $\|\bLambda\|_{\max} \leq C$.

\item[(iii)] There exist constants $C_{1}, C_{2} > 0$ such that $\lambda_{\min}(\bSigma_{u})>C_{1}$ and $\|\bSigma_{u}\|_{1} \leq C_{2}m_{p}$.

\item[(iv)] 
Let $d = T^{1-\alpha}$ for $\alpha \in (0,1)$.
The sample correlation matrix using $d$-day return data, $\widehat{\bR}_{h} = \widehat{\bD}_{h}^{-\frac{1}{2}}\widehat{\bSigma}_{h} \widehat{\bD}_{h}^{-\frac{1}{2}}$, where $\widehat{\bD}_{h}$ is the diagonal matrix consisting of the diagonal elements of $\widehat{\bSigma}_{h} = T^{-\alpha}\sum_{t=1}^{T^{\alpha}}(y_{t} - \bar{y})(y_{t} - \bar{y})'$, satisfies
\begin{equation*}
	\|\widehat{\bR}_{h}-\bR_{h}\|_{\max} = O_{P}(\sqrt{\log p /T^{\alpha}}).  \label{max norm}
\end{equation*}

\item[(v)]
Denote $\bSigma = (\Sigma_{ij})_{p \times p}$. The sample covariance matrix using $T$ observations, $\widehat{\bSigma} = T^{-1}\sum_{t=1}^{T}(y_{t} - \bar{y})(y_{t} - \bar{y})' = (\widehat{\Sigma}_{ij})_{p\times p}$, satisfies that, for  $s \in \{1,\dots, S\}$,
$$
\max_{\{i,j\} \in s}|\widehat{\Sigma}_{ij} - \Sigma_{ij}| = O_{P}(\sqrt{\log p /T}).
$$
\end{itemize}
\end{assum}

\begin{remark}
Assumption \ref{assum1}(i) is the conventional normalization condition in the factor model literature.
Assumption \ref{assum1}(ii) is known as the factor pervasiveness assumption, which is closely related to the incoherence structure \citep{fan2018eigenvector}. 
This assumption can hold in macroeconomic and financial applications and is used for analyzing low-rank matrices \citep{bai2003inferential, Chamberlain1983, fan2013large, fan2016incorporating,  lam2012factor, stock2002forecasting}. 
Specifically, we allow the global and local factors to be weak by imposing technical conditions on $a_{1}$ and $a_{2}$.
Intuitively, their lower bounds imply that both factors should have enough signals to satisfy the pervasive condition at different levels (see \citealp{choi2023large}).
Assumptions \ref{assum1}(iv)-(v) provide a high-level sufficient condition for analyzing large matrices.
The sample correlation matrix with $d$-day return data serves as the initial estimator for $\bR_{h}$ in Assumption \ref{assum1}(iv). 
This condition is required to account for the spillover effect between continents.
Here, the correlation matrix is considered to overcome the amplified scale issue of using the sample covariance matrix based on lower-frequency data (see Remark A.1 in the online supplement).
On the other hand, we can impose the element-wise convergence condition for each continent using the sample covariance matrix based on all observations (Assumption \ref{assum1}(iv)).
These element-wise convergence rate conditions are easily satisfied under the sub-Gaussian condition and mixing time dependency (\citealp{fan2018large, fan2018eigenvector, vershynin2010introduction, wang2017asymptotics}). 
It can also be satisfied under heavy-tailed observations with bounded fourth moments \citep{fan2017estimation, fan2018large, fan2018eigenvector, fan2021shrinkage}.
\end{remark}


We obtain the following convergence rates for Structured-POET and its inverse under various norms.
\begin{thm}\label{thm1}
	Suppose that $m_{p} = o(p^{c(5a_{2}-3)/2})$ and Assumption \ref{assum1} holds. 
	Let $\omega_{T} =p^{\frac{5}{2}(1-a_{1})+\frac{5}{2}c(1-a_{2})}\sqrt{\log p/T}+1/p^{\frac{5}{2}a_{1}-\frac{3}{2}+c(\frac{5}{2}a_{2}-\frac{7}{2})} +m_{p}/\sqrt{p^{c(5a_{2}-3)}}$. 
	If $m_{p}\omega_{T}^{1-q} = o(1)$, we have
				\begin{align}
							&\|\widehat{\bSigma}^{\mathcal{S}} - \bSigma\|_{\max} = O_{P}\left(\omega_{T}+p^{5(1-a_{1})}\Big(\sqrt{\frac{\log p}{T^{\alpha}}} + \frac{1}{T^{(1-\alpha)\beta}}\Big) + \frac{1}{p^{5a_{1}-4-c}}\right), \label{maxnorm_SPOET}\\
							& \|(\widehat{\bSigma}^{\mathcal{S}})^{-1} - \bSigma^{-1}\| = O_{P}\left(m_{p}\omega_{T}^{1-q}+p^{\frac{c}{2}(1-a_{2})}\omega_{T}+p^{\frac{11}{2}(1-a_{1})}\Big(\sqrt{\frac{\log p}{T^{\alpha}}} + \frac{1}{T^{(1-\alpha)\beta}}\Big) + \frac{1}{p^{\frac{11}{2}a_{1}-\frac{9}{2}-c}}\right).
							\label{inverse_SPOET}
				\end{align}
	In addition, if $a_{1}>\frac{3}{4}$ and $a_{2}>\frac{3}{4}$, we have
				\begin{align}
							\|\widehat{\bSigma}^{\mathcal{S}} - \bSigma\|_{\Sigma} =&O_{P}\Big(m_{p}\omega_{T}^{1-q} + p^{\frac{7}{2}(1-a_{1})}\big(\sqrt{\frac{\log p}{T^{\alpha}}} + \frac{1}{T^{(1-\alpha)\beta}}\big)  + \frac{1}{p^{\frac{7}{2}a_{1}-\frac{5}{2}-c}}\nonumber\\
							& \qquad\qquad +p^{\frac{21}{2}-10a_{1}}\big(\frac{\log p}{T^\alpha} +\frac{1}{T^{2(1-\alpha)\beta}}\big)+ \frac{1}{p^{10a_{1}-\frac{17}{2}-2c}} + \frac{m_{p}^{2}}{p^{5ca_{2}-3c - \frac{1}{2}}}\Big), \label{relativefro_SPOET}	
				\end{align}
    where the relative Frobenius norm is $\|\widehat{\bSigma}-\bSigma\|_{\Sigma} = p^{-1/2}\|\bSigma^{-1/2}\widehat{\bSigma}\bSigma^{-1/2}-\bI_{p}\|_{F}.$
\end{thm}

\begin{remark}\label{remark3.2}
	  $\omega_{T}$ is related to  the estimation of latent local factors and idiosyncratic components using $T$ observations.
	The additional terms $\sqrt{\log p/T^{\alpha}}$ and $1/T^{(1-\alpha)\beta}$ are the cost to handle the non-synchronized trading hours, when estimating latent global factors. 
	In particular, the first term is coming from sub-sampled observations, $T^{\alpha} = T/d$, while  the second term is the cost to estimate the synchronized correlation matrix $\bR_0$. 
	The optimal choice of $\alpha$ is $\alpha^{\ast} = \frac{2\beta}{1+2\beta}$, which simultaneously minimizes the convergence rates.
	This implies that the choice of frequency window $d$ is important in practice.
	In the numerical study, we use the weekly data, that is, $d=5$. 
\end{remark}

For simplicity, consider the case of strong global and local factors (i.e., $a_{1}=1$ and $a_{2} =1$),  $q=0$, and $m_{p} = O(1)$.
Then, the proposed Structured-POET method yields
			\begin{align*}
				\|\widehat{\bSigma}^{\mathcal{S}} - \bSigma\|_{\Sigma} =&O_{P}\left(\sqrt{\frac{\log p}{T^{\alpha}}} + \frac{1}{T^{(1-\alpha)\beta}}  + \frac{1}{p^{1-c} + p^{c}} +\sqrt{p}\Big(\frac{\log p}{T^\alpha} +\frac{1}{T^{2(1-\alpha)\beta}}\Big)+ \frac{1}{p^{\frac{3}{2}-2c}} + \frac{1}{p^{2c - \frac{1}{2}}}\right),
			\end{align*}
which can be convergent as long as $p = o(T^{\alpha})$ and $\frac{1}{4} < c < \frac{3}{4}$.
Similar to the Double-POET estimator \citep{choi2023large}, the upper and lower bounds of $c$ are required to estimate both global and national factor components.

To compare Structured-POET and Double-POET, we derive the convergence rate of Double-POET  (see Section A.2 in the online supplement). 
For simplicity, consider $m_{p} = O(1)$, $a_{1} =1$, and $a_{2}=1$, and ignore the log order terms.
Define the optimal $\alpha^{\ast} = \frac{2\beta}{1+2\beta}$ (see Remark \ref{remark3.2}).
With $\alpha = \alpha^{\ast}$, we have 
\begin{align*}
	\|\widehat{\bSigma}^{\mathcal{D}}-\bSigma\|_{\Sigma} &= O_{P}\left (\left (\frac{1}{T^{\frac{\beta}{1+2\beta}}}+ \frac{1}{p^{1-c}} + \frac{1}{p^{c}}\right)^{1-q} +  \frac{\sqrt{p}}{T^{\frac{2\beta}{1+2\beta}}}+ 
	\frac{1}{p^{\frac{3}{2}-2c}} +\frac{1}{p^{2c-\frac{1}{2}}} \right),\\
	\|\widehat{\bSigma}^{\mathcal{S}} - \bSigma\|_{\Sigma} 
	&=O_{P}\left(\left(\frac{1}{\sqrt{T}} + \frac{1}{p^{1-c}} + \frac{1}{p^{c}}\right)^{1-q} + \frac{1}{T^{\frac{\beta}{1+2\beta}}}  +\frac{\sqrt{p}}{T^{\frac{2\beta}{1+2\beta}}} + \frac{1}{p^{\frac{3}{2}-2c}} + \frac{1}{p^{2c - \frac{1}{2}}}\right),
\end{align*}
where $\widehat{\bSigma}^{\mathcal{D}}$ is the Double-POET estimator defined in Section A.1 of the online supplement.
Specifically, when $q \neq 0$, Structured-POET achieves a faster convergence rate under the relative Frobenius norm.
This is because utilizing all observations enhances the estimation accuracy of each block diagonal matrix.
However, when $q = 0$, the convergence rates of both estimators are the same.
This is because the estimation error of the correlations between continents dominates the benefit mentioned above.
Importantly, we note that this does not mean that their estimation errors are exactly the same. 
In fact, based on our simulation study, we can conjecture that Structured-POET has smaller convergence rates than Double-POET for $q=0$.
That is, the relative ratio of the convergence rate of Structured-POET with respect to that of Double-POET may be less than 1. 
Unfortunately, due to the complex upper bound calculations used to handle high-dimensional matrices, we cannot theoretically  show this statement for $q=0$. 
We leave this for a future study. 
Similarly, under the spectral norm for the inverse matrix, the convergence rate of Structure-POET can be faster than that of Double-POET when $q \neq 0$.

\section{Simulation Study} \label{simulation}

In this section, simulations are carried out to examine the finite sample performance of the proposed Structured-POET method. 
We considered the true covariance as $\bSigma = \bB\bB' + \bLambda\bLambda '+ \bSigma_{u}$, where
each row of $\bB$ was drawn from $\mathcal{N}(\mu_{B},\bI_{k})$, where each element of $\mu_{B}$ is i.i.d. Uniform$(-0.5,0.5)$; for $\bLambda = \diag(\Lambda^{1},\dots, \bLambda^{l})$, each row of $\bLambda^{l}$ for each $l$ was drawn from $\mathcal{N}(\mu_{\bLambda^{l}}, \bI_{r_{l}})$, where each element of $\mu_{\Lambda^l}$ is i.i.d. Uniform$(-0.3,0.3)$. 
We generated $\bSigma_{u}$ as follows.
Let $\bD_{u} = \diag(d_{1},\dots,d_{p})$, where each $\{d_{i}\}$ was generated independently from Uniform$(0.5,1.5)$. 
Let $\pi = (\pi_{1},\dots, \pi_{p})'$ be a sparse vector, where each $\pi_{i}$ was drawn from $\mathcal{N}(0,1)$ with probability $\frac{0.5}{\sqrt{p}\log{p}}$, and $\pi_{i} = 0$ otherwise. 
Then, we set $\bSigma_u = \bD_{u} + \pi\pi' - \diag\{\pi_{1}^{2},\dots,\pi_{p}^{2}\}$. 
In the simulation, we generated $\bSigma_u$ until it was positive definite.

Let $\bD$ be the diagonal matrix consisting of the diagonal elements of $\bSigma$. 
We then obtained the true correlation matrix $\bR = \bD^{-\frac{1}{2}}\bSigma \bD^{-\frac{1}{2}} = (\rho_{0,ij})_{p \times p}$.
Next, we set $\bR_{h} =  (\rho_{h,ij})_{p \times p}$, where $\rho_{h,ij} = \sgn(\rho_{0,ij})(|\rho_{0,ij}| + 0.5h^{\beta})$ if $i$ and $j$ belong to different continent groups, for $h = \frac{0.5}{d}$ and $\beta = 0.75$, and $\rho_{h,ij} = \rho_{0,ij}$ if $i$ and $j$ are in the same continent group.
Let $\{\gamma_{i}, v_{i}\}_{i=1}^{k}$ be the leading eigenvalues and eigenvectors of $\widetilde{\bSigma}_{g} = \bD^{\frac{1}{2}}\bR_{h}\bD^{\frac{1}{2}} - \bLambda\bLambda '- \bSigma_{u}$.
Then, we obtained $\bB_{h} = \bV\bGamma^{\frac{1}{2}}$, where $\bGamma = \diag(\gamma_{1},\dots,\gamma_{k})$ and $\bV = (v_{1},\dots,v_{k})$.
We note that $\bB_{h}$ represents the non-synchronized structure.
Thus, we generated non-synchronized observations by
$$
y_{t} = \bB_{h}G_{t} + \bLambda F_{t} + u_{t},
$$
where $G_{t} = \bUpsilon G_{t-1} + \upsilon_{t}$, $F_{t} = \bar{\bUpsilon} F_{t-1} + \bar{\upsilon}_{t}$, and $u_{t} = \bSigma_{u}^{1/2}\widetilde{u}_{t}$, where $\widetilde{u}_{t} = \widetilde{\bUpsilon} \widetilde{u}_{t-1} + \epsilon_{t}$, with $k \times k$, $r\times r$, $p \times p$ diagonal matrices $\bUpsilon$, $\bar{\bUpsilon}$, and $\widetilde{\bUpsilon}$, respectively. 
Each diagonal element of $\bUpsilon$, $\bar{\bUpsilon}$, and $\widetilde{\bUpsilon}$ was generated from Uniform(0,0.7), and $\upsilon_{t}$, $\bar{\upsilon}_{t}$, and $\epsilon_{t}$ were drawn from $\mathcal{N}(0,\bI_{k})$,  $\mathcal{N}(0,\bI_{r})$, and  $\mathcal{N}(0,\bI_{p})$, respectively.

In this simulation study, we fixed the number of individuals $p = 500$.
We set the number of continents $S=2$ and the number of local groups $L=20$ such that each continent group included $10$ local groups (i.e., $p_l = 25$).
Also, we chose the numbers of factors as $k=3$ and $r = L\times r_{l}$, where $r_{l}=2$ for each local group $l$.
Then, we considered two cases: (i) increasing $T$ from 100 to 600 in increments of 50 with the size of frequency $d \in \{1,5\}$ (i.e., in-sample size is $T/d$) and (ii) increasing $d$ from 1 to 10 with a fixed $T=600$.
For each case, 200 simulations were conducted.

For comparison, the sample covariance matrix (SamCov), POET, Double-POET (D-POET), and Structured-POET (S-POET) methods were employed to estimate $\bSigma$. 
The average estimation errors were measured under the following norms: $\|\widehat{\bSigma} -\bSigma\|_{\Sigma}$, $\|\widehat{\bSigma} -\bSigma\|_{\max}$, and $\|(\widehat{\bSigma})^{-1} - \bSigma^{-1}\|$, where $\widehat{\bSigma}$ is one of the covariance matrix estimators. 
We note that the lower-frequency data is obtained by summing the observations with $d$-day window. 
Therefore, for the SamCov estimator and the initial pilot estimator of POET and D-POET, we used $d^{-1}\widehat{\bSigma}_{h}$ (see Section A.2 in the online supplement). 
For each estimation, we determined the number of factors for D-POET and POET using the eigenvalue ratio methods suggested by \cite{choi2023large} and \cite{ahn2013eigenvalue}, with $k_{\max}=10$ and $r_{l,\max}=10$, respectively.
For the S-POET estimation, we chose the number of ranks for each off-diagonal block by the largest singular value ratio.
In addition, we employed the soft thresholding scheme for the idiosyncratic covariance matrix estimation.

\begin{figure}
	\includegraphics[width=\linewidth]{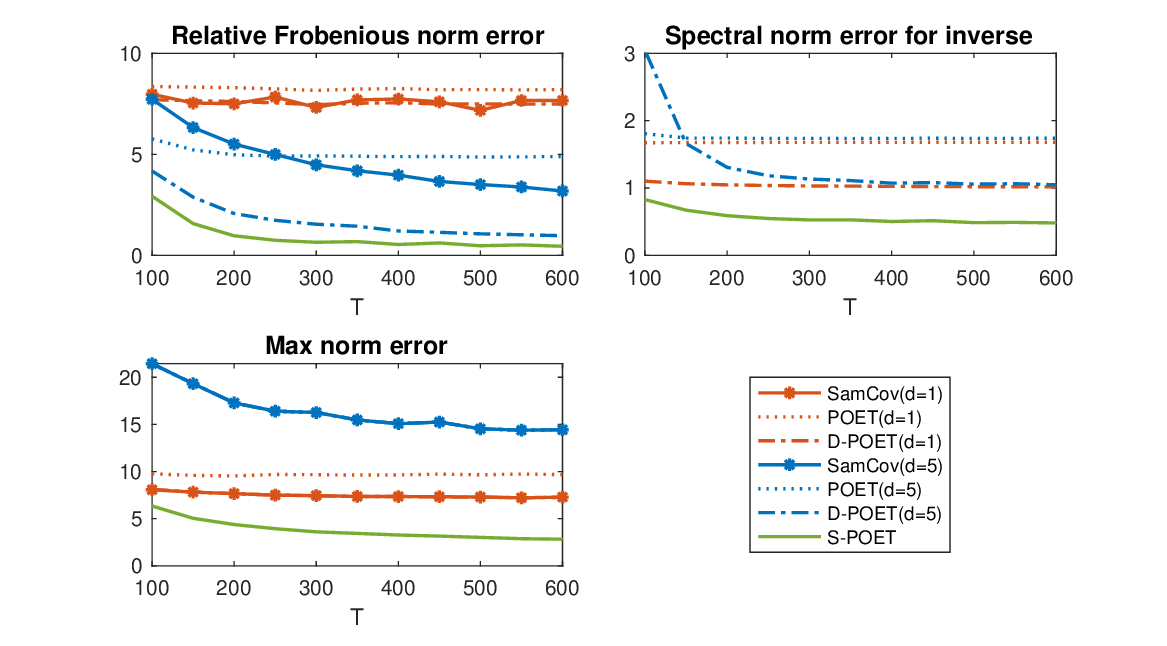}
	\centering	
	\caption{Averages of $\|\widehat{\bSigma} -\bSigma\|_{\Sigma}$,  $\|(\widehat{\bSigma})^{-1} - \bSigma^{-1}\|$, and  $\|\widehat{\bSigma} -\bSigma\|_{\max}$ for SamCov, POET, D-POET, and S-POET against $T$ with fixed $p=500$ and $L=20$. 
    Lines that exceed the upper limits of the y-axis are excluded for the spectral norm error plot.
    }				\label{varyingT}
\end{figure}

\begin{figure}
	\includegraphics[width=\linewidth]{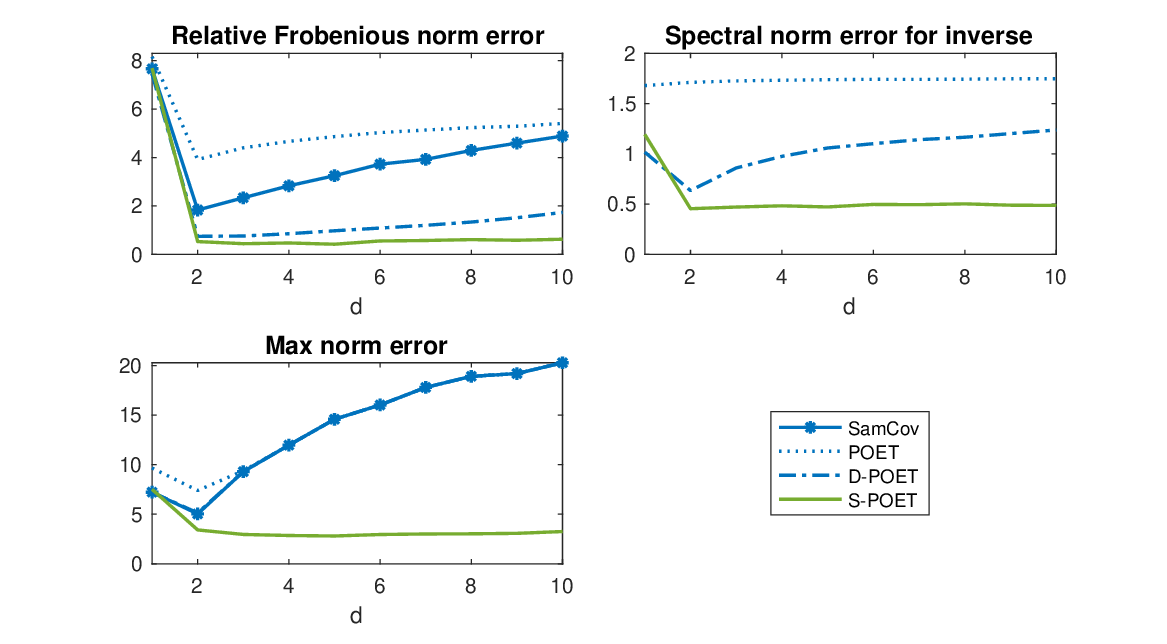}
	\centering	
	\caption{Averages of $\|\widehat{\bSigma} -\bSigma\|_{\Sigma}$,   $\|(\widehat{\bSigma})^{-1} - \bSigma^{-1}\|$, and $\|\widehat{\bSigma} -\bSigma\|_{\max}$ for SamCov, POET, D-POET, and S-POET against $d$ with fixed $p=500$, $T = 600$, and $L=20$.
    Lines that exceed the upper limits of the y-axis are excluded for the spectral norm error plot.
    }				\label{varyingd}
\end{figure}

Figures \ref{varyingT} and \ref{varyingd} depict the averages of estimation errors under different norms against $T$ and $d$, respectively.
From Figures \ref{varyingT} and \ref{varyingd}, we find that S-POET has smaller estimation errors than the other methods under different norms. 
Specifically, in Figure \ref{varyingT}, as $T$ increases, the estimation errors of S-POET and estimators with $d=5$ decrease, while the estimation errors of estimators with $d=1$ do not decrease.
This is because the estimators with $d=1$ actually estimate $\bSigma_h$ not $\bSigma_0$, and when $d=1$, $\bSigma_h$ is not close to $\bSigma_0$. 
When comparing the estimation procedures with $d=5$, S-POET shows the best performance. 
This is because S-POET can accurately estimate the local factors and idiosyncratic components by utilizing whole observations, while other estimators utilize lower-frequency observations, which causes inefficiency. 
Figure \ref{varyingd} indicates that the estimation errors of all methods dramatically drop from $d=1$ to $d=2$, which is consistent with the results shown in Figure \ref{varyingT}.
S-POET shows stable results and has the minimum estimation errors when $d=5$.
In contrast, as the frequency size $d$ increases, the estimation errors of SamCov, POET, and D-POET tend to increase again due to the smaller sample sizes.
That is, the loss of information is severe only when using  lower-frequency  observations.
From this result, we can conjecture that for a fixed $T$, the estimation error resulting from a small sample size with larger $d$ is greater than the error resulting from the effect of observation time gaps.
However, the proposed S-POET incorporates all available data to estimate the same regional covariance matrices, which helps enjoy the efficiency.  
It is worth noting that the estimation errors of POET seem constant as $d$ increases.
This is because POET does not estimate the local covariance matrix, which may dominate other estimation errors.
The above results support  the theoretical findings established in Section \ref{asymp}.

\section{Empirical Study} \label{empiric}
We conducted a minimum variance portfolio allocation study using the proposed Structured-POET method with global financial data.
We obtained the daily transaction prices of international stock markets over 15 countries by the total market capitalization.
The whole sample period is from January 3, 2017, to December 30, 2022. 
After excluding stocks with missing returns and no variation, we picked $1500$ stocks for this period based on the market cap for each country.
In particular, we selected $500$ firms for each continent and calculated both daily and weekly log-returns. 
The distribution of our sample is presented in Table A.1 in the online supplement.

We computed several estimators, including the S-POET, D-POET, POET, and SamCov estimators, for each week. 
For S-POET, we used weekly or two-day window returns to estimate the global factor component and daily returns to estimate the local factor and idiosyncratic components. 
We employed all daily, two-day window, and weekly returns for the other procedures.
For all POET-type procedures, we estimated the idiosyncratic volatility matrix using information of the 11 Global Industrial Classification Standard (GICS) sectors  \citep{ait2017using, fan2016incorporating}. 
Specifically, we set the idiosyncratic components to zero for the different sectors, while maintaining them for the same sector. 
For a robustness check, we used different numbers of global factors, $k$, ranging from 1 to 5 for D-POET and from 1 to 20 for POET.
For D-POET, we chose the number of local factors using the eigenvalue ratio method proposed by \cite{ahn2013eigenvalue} with $r_{l, \max} = 10$.
In the S-POET procedure, for each off-diagonal partitioned block, we used the best rank-one approximation as suggested by the largest singular value gap method (see Remark \ref{rankchoice}).

We considered the following constrained minimum variance portfolio allocation problem \citep{fan2012vast, jagannathan2003risk} to analyze the out-of-sample portfolio allocation performance:
\begin{equation*}\label{min problem}
	\min_{\omega} \omega^{T}\widehat{\bSigma}\omega, \text{ subject to } \omega^{\top}\mathbf{1} = 1, \; \|\omega\|_{1}\leq c,
\end{equation*}
where $\mathbf{1} = (1,\dots,1)^{\top} \in \mathbb{R}^{p}$, the gross exposure constraint $c$ varies from 1 to 4, and $\widehat{\bSigma}$ is one of the volatility matrix estimators obtained from S-POET, D-POET, POET, and SamCov. 
At the beginning of each week, we obtained optimal portfolios based on each estimator using the past 12 months' returns and held these portfolios for one week.
We then computed the square root of the realized volatility using the weekly log-returns. 
Their averages were recorded for the out-of-sample risk. 
We examined six out-of-sample periods: 2018, 2019, 2020, 2021, 2022, and the full period from 2018 to 2022.

\begin{figure}[t]
	\includegraphics[width=\linewidth]{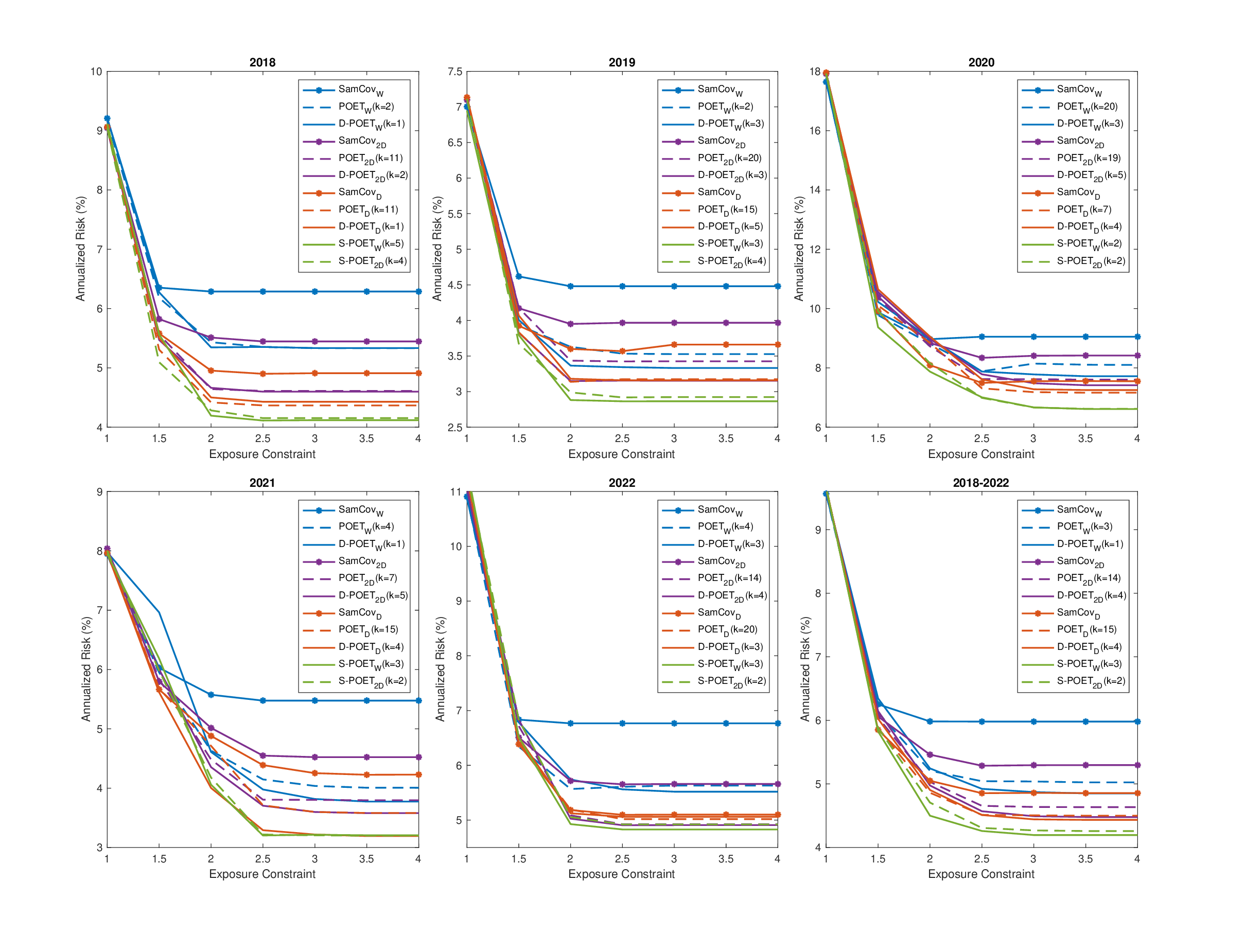}
	\centering
	\caption{Out-of-sample risks of the optimal portfolios constructed by the SamCov, POET, Double-POET, and Structured-POET estimators for the global stock market.}				\label{ALLPLOT_1yr}
\end{figure}

Figure \ref{ALLPLOT_1yr} illustrates the out-of-sample risks of the portfolios constructed by SamCov, POET, D-POET, and S-POET under varying exposure constraints.
To draw readable plots, we presented the best performing results among the range of $k$ for each estimator type and each period.
We used subscripts to explicitly denote the frequency of the data used, with W (blue lines), 2D (purple lines), and D (red lines) representing weekly, two-day window, and daily data, respectively.
As shown in Figure \ref{ALLPLOT_1yr}, S-POET consistently outperforms the other estimators.
Specifically, for all periods except 2021, S-POET$_{\text{W}}$ reduces the minimum risks by 4.7\%--10.2\% compared to the best estimator among the other methods.
When comparing the estimation procedures with the same frequency observations, D-POET exhibits lower risks than POET and SamCov.
Furthermore, D-POET, POET, and SamCov estimators using daily data tend to have lower risks than those using weekly data.
This may be because, in practice, the impact of estimation inefficiency resulting from the smaller sample size could be greater than that resulting from the different observation time points.
In addition, we also calculated the out-of-sample Sharpe ratio and averaged return in Table \ref{Sharpe}.
The results are based on selected estimators with minimum variances for the full period from 2018 to 2022.
The results also indicate that S-POET$_{\text{W}}$ outperforms the other estimators.
In summary, for portfolio allocation in the global stock market, S-POET, which incorporates both daily and weekly returns under the specific observation structure, outperforms POET and Double-POET, which use only daily, two-day window, or weekly returns.
From this result, we can conjecture that the estimation accuracy for the national factor and idiosyncratic components can be improved using more frequent data (i.e., daily returns) for each continent group.
In addition, for global factor estimation, using less frequent data (i.e., weekly returns) can manage the  different trading hour problem.

\begin{table}[htbp]
\centering
\caption{Out-of-sample Sharpe ratios and returns (multiplied by $10^4$) for the full period from 2018 to 2022} \label{Sharpe}
\label{tab:my-table}
\begin{tabular}{@{}lllllll@{}}
\toprule
& SamCov$_{\text{W}}$ & POET$_{\text{W}}$ & D-POET$_{\text{W}}$ & SamCov$_{\text{2D}}$ & POET$_{\text{2D}}$  & D-POET$_{\text{2D}}$\\ \hline
Sharpe ratio & 0.0477 & 0.0521 & 0.0580 &0.0637 & 0.0579 & 0.0658\\
Return &6.255 &  6.129 &  6.771 & 7.795 & 6.716 &  7.410\\
\midrule
& SamCov$_{\text{D}}$ & POET$_{\text{D}}$  & D-POET$_{\text{D}}$ & S-POET$_{\text{W}}$ & S-POET$_{\text{2D}}$ &\\ \hline
Sharpe ratio & 0.0712 & 0.0656 & 0.0663 & 0.0841  & 0.0696&  \\
Return & 7.978 & 7.315 & 7.373 & 8.965 & 7.417&  \\ \bottomrule
\end{tabular}
\end{table}

\section{Conclusion} \label{conclusion}

In this paper, we introduce a novel large global volatility matrix inference procedure. 
The proposed Structured-POET method leverages observation structural information from global financial markets based on latent global and national factor models.
We establish the asymptotic properties of Structured-POET and demonstrate its efficiency in estimating a large global covariance matrix compared to the Double-POET procedure. 
In our empirical study, we demonstrate that the proposed Structured-POET estimator outperforms the other existing methods in minimum variance portfolio allocation problems. 
This is because the Structured-POET procedure accurately estimates the latent national factors and idiosyncratic components using daily returns. Additionally, using weekly returns to estimate the latent global factors can mitigate the effect of different trading hours across markets.
Overall, our findings support the effectiveness of the Structured-POET method in estimating large global volatility matrices. 

On the estimation procedure, we simplify the inter-continent correlation on the idiosyncratic term to reduce the estimation error because of the loss of information issue.
However, it is interesting and important to handle both estimation error and model specification error, and thus, we leave this for a future study.

\bibliography{S-POET}

	\end{document}